%% file: main.tex
\documentclass{article}

\usepackage[preprint]{jmlr2e}
\usepackage{natbib}

\usepackage[utf8]{inputenc} 
\usepackage[T1]{fontenc}    
\usepackage{hyperref}       
\usepackage{url}            
\usepackage{booktabs}       
\usepackage{amsfonts}       
\usepackage{nicefrac}       
\usepackage{microtype}      
\usepackage{xcolor}         
\usepackage{booktabs} 

\usepackage{subcaption}

\usepackage{times}
\usepackage{dblfloatfix}

\input{PackageHeader2024.tex}
\input{Notations_Quantum22.tex}
\input{acronyms.tex}

\def\STATE{\State}
\def\FOR{\For}
\def\ENDFOR{\EndFor}

\def\cite{\citep}

\settoggle{Track}{true}

\usepackage{ulem} 

\title{Quantum Shadow Gradient Descent for Variational Quantum Algorithms}
\author{
\name     Mohsen Heidari\textsuperscript{\rm 1},
    Mobasshir A Naved\textsuperscript{\rm 2},
    Zahra Honjani\textsuperscript{\rm 1},
    Wenbo Xie\textsuperscript{\rm 2},
      Arjun Jacob Grama\textsuperscript{\rm 2},
          Wojciech Szpankowski\textsuperscript{\rm 2},\\
\addr     \textsuperscript{\rm 1}Department of Computer Science, Indiana University, Bloomington\\
    \textsuperscript{\rm 2}Department of Computer Science, Purdue University, West Lafayette\\
    \{mheidar, zhonjani\} @iu.edu, \{naved, xie401, agrama, szpan\} @purdue.edu
       }
       
\begin{document}

\maketitle

\begin{abstract}
Gradient-based optimizers have been proposed for training variational quantum circuits in settings such as quantum neural networks (QNNs). The task of gradient estimation, however, has proven to be challenging, primarily due to distinctive quantum features such as state collapse and measurement incompatibility. Conventional techniques, such as the \textit{parameter-shift rule}, necessitate several fresh samples in each iteration to estimate the gradient due to the stochastic nature of state measurement. 
Owing to state collapse from measurement, the inability to reuse samples in subsequent iterations motivates a crucial inquiry into whether fundamentally more efficient approaches to sample utilization exist.  In this paper, we affirm the feasibility of such efficiency enhancements through a novel procedure called quantum shadow gradient descent (QSGD), which uses a {\em single sample per iteration} to estimate {\em all} components of the gradient. Our approach is based on an adaptation of shadow tomography that significantly enhances sample efficiency. Through detailed theoretical analysis, we show that QSGD has a significantly faster convergence rate than existing methods under locality conditions.  
We present detailed numerical experiments supporting all of our theoretical claims.

\end{abstract}

\input{intro}

\input{model}

\input{ShadowGradient.tex}

\input{analysis}
\input{Numerical}

\input{RelatedWorks}
\section*{Conclusion}
We proposed a novel quantum shadow sampling technique to rapidly estimate quantum states while being consistent with the no-cloning postulate. Our one-shot method takes quantum states as input and produces several shadows for each state. We use the shadows to approximate the gradient of the loss with respect to the ansatz parameters. Based on our shadowing procedure, we proposed  a novel update rule called quantum shadow gradient descent. 
We proved that our method has excellent convergence guarantees. 
In summary, we presented a novel gradient-based update that is one-shot, has a fast convergence rate and is applicable to a large class of unitaries. We supported all our theoretical claims with experimental verification.


\bibliographystyle{alpha}
\bibliography{main}


\appendix
\onecolumn

\input{proof_lem_loss_derivative.tex}

\section{Proof of Lemma \ref{lem:gradient unbiased}}\label{proof:lem:gradient unbiased}
\begin{customlemma}{\ref{lem:gradient unbiased}}
Let $(\hat{y}, z)$ be the output of the circuit in Fig. \ref{fig:loss derivative circuit multi} and  $g_l :=-2(-1)^z \ell(y, \hat{y})$. Then, $\EE[g_l] = \frac{\partial  \Loss( \parameter, \ket{\phi}, y).
}{\partial a_{\bfs_l}}$
\end{customlemma}
\begin{proof}
It is not difficult to see that the circuit in Figure \ref{fig:loss derivative circuit multi}  outputs:  
\begin{align*}
  \tilde{\rho}_l:=\frac{1}{2}\Big( & R_{\bfs_l}\qty(\pi/2)  \rho^{out}_l  R_{\bfs_l}\qty(-\pi/2) \tensor \ketbra{0} + R_{\bfs_l}\qty(\pi/2)  \rho^{out}_l  R_{\bfs_l}\qty(\pi/2) \tensor \outerproduct{0}{1}\\
 & + R_{\bfs_l}\qty(-\pi/2)  \rho^{out}_l  R_{\bfs_l}\qty(-\pi/2) \tensor \outerproduct{1}{0}+ R_{\bfs_l}\qty(-\pi/2)  \rho^{out}_l  R_{\bfs_l}\qty(\pi/2) \tensor \outerproduct{1}{1}\Big).
\end{align*}
After the measurements with $(\hat{y},z)$ as the outcome, we output $g_{\bfs_l} =-2 (-1)^z \ell(y, \hat{y})$. Taking the expectation over $(Z, \hat{Y})$ gives
\begin{align*}
    \EE[g_{\bfs_l}] &=-2 \sum_{\hat{y}, z} (-1)^z \ell(y_t,\hat{y}) \tr{M_{\hat{y}}\tensor\ketbra{z} \tilde{\rho}_l }\\
    & =-2 \sum_{\hat{y}, z} (-1)^z \ell(y_t,\hat{y}) \tr{M_{\hat{y}} \tr_{Z}\{ \ketbra{z} \tilde{\rho}_l} \},
   \end{align*}
where $\tr_{Z}$ is the partial trace over the last qubit of $\tilde{\rho}_l$. This term equals to 
\begin{align*}
\tr_{Z}\{ \ketbra{0} \tilde{\rho}_l\} & = \frac{1}{2}\qty(R_{\bfs_l}\qty(+\pi/2)  \rho^{out}_l  R_{\bfs_l}\qty(-\pi/2))\\
\tr_{Z}\{ \ketbra{1} \tilde{\rho}_l\} & = \frac{1}{2}\qty( R_{\bfs_l}\qty(-\pi/2)  \rho^{out}_l  R_{\bfs_l}\qty(+\pi/2))
\end{align*}
Therefore, we have that
\begin{align}\label{eq:EE g_sl}
 \EE[g_{\bfs_l}] & = -\sum_{\hat{y}} \ell(y_t,\hat{y}) \tr\Big\{ M_{\hat{y}} \Big( R_{\bfs_l}\qty(+\pi/2)  \rho^{out}_l  R_{\bfs_l}\qty(-\pi/2) -  R_{\bfs_l}\qty(-\pi/2)  \rho^{out}_l  R_{\bfs_l}\qty(+\pi/2)\Big)\Big\}
\end{align}
Next, we show that this is equal to the partial derivative of $\Loss$. From Lemma \ref{lem:loss derivative}, the derivative equals to the following
\begin{align*}
\frac{\partial  \Loss( \parameter, \ket{\phi}, y)
}{\partial a_{\bfs_l}} & = \sum_{\hat{y}} i \ell(y_t,\hat{y})  \tr{M_{\hat{y}} U_{>l} \big[ \sigma^{\bfs_l}, {\rho}_l^{out}\big]U^\dagger_{>l}}.
\end{align*}
It is known that (see \citep {Mitarai2018}) the commutator of a Pauli string $\qps$ with any operator $A$ equals to : 
\begin{align*}
\big[ \qps, A\big] &= i \Big( e^{-i\frac{\pi}{4}\qps}Ae^{i\frac{\pi}{4}\qps} - e^{i\frac{\pi}{4}\qps}Ae^{-i\frac{\pi}{4}\qps} \Big).
\end{align*}
With this relation, the derivative of the loss equals:
\begin{align*}
\frac{\partial  \Loss( \parameter, \ket{\phi}, y)
}{\partial a_{\bfs_l}} = -\sum_{\hat{y}}  \ell(y_t,\hat{y})  &\tr\Big\{ M_{\hat{y}} \Big( e^{-i\frac{\pi}{4}\sigma^{\bfs_l}}\rho_l^{out}e^{i\frac{\pi}{4}\qps} - e^{i\frac{\pi}{4}\qps}\rho_t^{out}e^{-i\frac{\pi}{4}\qps} \Big)\Big\}.
\end{align*}
This expression is equal to the right hand side of \eqref{eq:EE g_sl}. Hence the proof is completed. 
\end{proof}

\input{Proof_thm_shadow_unbiased}

\input{convergence}

\section{Accuracy Upper Limit}\label{app:optimal loss}
We compare the performance the gradient-based trainers with the optimal expected loss within each batch. This lower-bound is obtained by Holveo-Holstrom theorem \citep{Holevo2012}. More precisely it is given in the following lemma which was proved in \cite{HeidariAAAI2022}.

\begin{lemma}\label{lem:optimal loss}
Given a set of $m$ labeled density operators $\{(\rho_j, y_j)\}_{j=1}^m$ with $y_j \in \pmm$, the minimum of the expected sample loss averaged over the $m$ samples equals:
$\frac{1}{2}\Big(1-\Big\| \frac{1}{m} \sum_j y_j \rho_j\Big\|_1\Big),$ 
where $\norm{\cdot}_1$ is the trace norm and the minimum is taken over all measurements applied on each $\rho_j$ for predicting $y_j$. 
\end{lemma}

\section{Convergence for strongly-convex loss}\label{app:strongly convex}
 To provide an insight into the rate of convergence of our proposed work as compared to existing approaches such as the parameter-shift rule and RQSGD, we consider the strong-convexity assumption for all methods. Strong-convexity is satisfied if there exists a constant $\alpha>0$ such that for all $\parameter$ and $\parameter'$:
\begin{align*}
\Loss(\parameter')\geq \Loss(\parameter) + \nabla\Loss(\parameter)^T(\parameter'-\parameter) + \frac{\alpha}{2} \norm{\parameter'-\parameter}_2^2.
\end{align*}
We present our convergence guarantees in the following, where the proof can be found in \cite{Bottou2018}. 
\begin{theorem}\label{thm:convergence stongly convex}
Under the strong-convexity assumption for $\Loss(\parameter)$, suppose that a SGD method with an update rule $\parameter^{(t+1)} = \parameter^{(t)} - \eta_t g_t,$ where the step size is $\eta_t=\frac{\beta}{\gamma+t}$ with constants $\beta >\frac{1}{\alpha}, \gamma>0$, and $\EE[\norm{g_t}_2^2]\leq \mathfrak{M} + \mathfrak{M}_G \norm{\nabla\Loss(\parameter^{(t)})}_2^2$ for some constants $\mathfrak{M}, \mathfrak{M}_G$. Then, 
\begin{align*}
\EE[\Loss(\parameter^{(T)})] - \Loss(\parameter^*) \leq \frac{\nu}{\gamma+T},
\end{align*}
where $T$ is the number of iterations, and 
\begin{align*}
\nu :=\max\big\{\frac{\beta^2 \mathfrak{L}~ \mathfrak{M}}{2(\beta \alpha-1)}, (\gamma+1)(\Loss(\parameter^{(1)})-\Loss(\parameter^*))\big\}.
\end{align*}
\end{theorem}
As noted in the classical setting, only the first term in $\nu$ matters when an appropriate initialization is used in the update rule. That brings us to the following quantity $\frac{\beta^2 \mathfrak{L}~ \mathfrak{M}}{2(\beta \alpha-1)(\gamma+T)}$ that behaves as the optimality gap in \eqref{eq:optimality gap} if we ignore the constants $\alpha, \beta$ and $\gamma$. Hence, we immediately conclude that QSGD has a faster convergence rate for a local ansatz.

\section{Implementation Details} \label{app:imp details}
The QNN framework is implemented using Pytorch with CUDA support to speed up learning using GPUs. The experiments are simulated in a classical computer.

\textbf{Training Details}. We had three hyper parameters to tune for when training the learning algorithm for the VQA's. To guarantee same sample utilization among above-mentioned three gradient based methods, the sample size was fixed for each of the experiments. We used the sample size of $57600$ for Exp. $1$ and Exp. $2$, and $61440$ for Exp. $3$. These sample sizes were determined in order to deal with varying numbers of optimizable parameters while ensuring an equal distribution of samples. The epoch ($300$ for Exp. $1$ and Exp. $2$; $60$ for Exp. $3$) and batchsize ($192$ for Exp. $1$ and Exp. $2$; $1024$ for Exp. $3$) hyperparameters, which determines how many times the learning algorithm will run through the entire training dataset and the number of samples to work through per iteration respectively were derived from the adopted fixed sample sizes. The most significant hyperparameter, step size, is optimized using the Grid search approach and leveraging the binary search algorithm concept for each experiment and underlying methods distinctly.

\textbf{Testing Measures.} We have generated a validation set of 2000 samples to test and compare the accuracy of all the training procedures for three different type of dataset and two distinctive architectural setup of quantum perceptrons(neurons).


\end{document}

%% file: PackageHeader2024.tex
\usepackage{amsmath}
\usepackage{amsfonts}
\usepackage{mathrsfs}
\usepackage{tikz}
\usepackage[utf8]{inputenc}

\usepackage{epstopdf}
\usepackage{mathtools}
\usepackage{graphicx}
\usepackage{enumerate}
\usepackage{epsf,verbatim,amssymb,array,cite,multicol,multirow}  
\usepackage{psfrag,bm,xspace}
\usepackage{hhline}
\usepackage{xstring}
\usepackage{ifthen}
\usepackage{booktabs}

\usepackage{xparse}
\usepackage{physics}

\usepackage{tikz}
\usetikzlibrary{quantikz}

\usepackage{algorithm}      
\usepackage[noend]{algpseudocode} 

\ifdefined \theorem 
\else
  \newtheorem{theorem}{Theorem}
\fi

\ifdefined \lemma 
\else
  \newtheorem{lemma}{Lemma}
  
\fi

\ifdefined \corollary 
\else
\newtheorem{corollary}{Corollary}
\fi

\ifdefined \proposition 
\else
\newtheorem{proposition}{Proposition}
\fi

\ifdefined \definition 
\else
\newtheorem{definition}{Definition}
\fi

\ifdefined \example 
\else
\newtheorem{example}{Example}
\fi

\ifdefined \remark 
\else
\newtheorem{remark}{Remark}
\fi

\ifdefined \conjecture 
\else
\newtheorem{conjecture}{Conjecture}
\fi

\newtheorem{innercustomgeneric}{\customgenericname}
\providecommand{\customgenericname}{}
\newcommand{\newcustomtheorem}[2]{%
  \newenvironment{#1}[1]
  {%
   \renewcommand\customgenericname{#2}%
   \renewcommand\theinnercustomgeneric{##1}%
   \innercustomgeneric
  }
  {\endinnercustomgeneric}
}

\newcustomtheorem{customthm}{Theorem}
\newcustomtheorem{customlemma}{Lemma}

\makeatletter
\def\old@comma{,}
\catcode`\,=13
\def,{%
  \ifmmode%
    \old@comma\discretionary{}{}{}%
  \else%
    \old@comma%
  \fi%
}
\makeatother

\usepackage{xcolor}
\definecolor{darkblue}{rgb}{0.1,0.1,0.8}
\definecolor{DarkGreen}{rgb}{0,0.6,0}
\definecolor{brickred}{rgb}{0.8, 0.25, 0.33}
\definecolor{britishracinggreen}{rgb}{0.0, 0.26, 0.15}
\definecolor{calpolypomonagreen}{rgb}{0.12, 0.3, 0.17}
\definecolor{ao(english)}{rgb}{0.0, 0.5, 0.0}
	\definecolor{cadmiumgreen}{rgb}{0.0, 0.42, 0.24}
\definecolor{burgundy}{rgb}{0.5, 0.0, 0.13}
\usepackage{etoolbox}

\providetoggle{Blue_revision}
\settoggle{Blue_revision}{true}

\providetoggle{Track}
\settoggle{Track}{true}
\newcommand{\addv}[3]{%
	\iftoggle{Track}{%
    	\IfEqCase{#1}{%
       	 	{a}{\ifthenelse{\equal{#2}{ON}}{{\color{cadmiumgreen}#3}}{#3}}%
        	{b}{\ifthenelse{\equal{#2}{ON}}{{\color{brickred}#3}}{#3}}%
       		{c}{\ifthenelse{\equal{#2}{ON}}{{\color{burgundy}#3}}{#3}}%
    	}[\PackageError{tree}{Undefined option to tree: #1}{}]%
	}{#3}%
}

 \usepackage{hyperref}
\usepackage{xcolor}

\hypersetup{
colorlinks=true, %
 pdfstartview={FitH},
   linkcolor=red,
   citecolor=blue,
   urlcolor={blue!80!black}
}

\usepackage{graphicx}

\newcounter{relctr} 
\everydisplay\expandafter{\the\everydisplay\setcounter{relctr}{0}} 

\AtBeginDocument{} 

\global\long\def\RR{\mathbb{R}}

\global\long\def\NN{\mathbb{N}}

\global\long\def\EE{\mathbb{E}}

\global\long\def\11{\mathbbm{1}}

\newcommand{\bfa}{\mathbf{a}}
\newcommand{\bfb}{\mathbf{b}}

\newcommand{\bfs}{\mathbf{s}}

\newcommand{\CM}{\mathcal{M}}

\global\long\def\+{\oplus}

\newcommand\pmm{\{-1,1\}}

\def\<{\langle}
\def\>{\rangle}

\ifdefined \var 
  \renewcommand{\var}{\mathsf{var}}
\else
  \newcommand{\var}{\mathsf{var}}
\fi

\ifdefined \abs \else
 \newcommand{\abs}[1]{\lvert#1\rvert}
\fi

\ifdefined \norm \else
 \newcommand{\norm}[1]{\lVert#1\rVert}
\fi

\ifdefined \set 
  \renewcommand{\set}[1]{\left\{#1\right\}}
\else
  \newcommand{\set}[1]{\left\{#1\right\}}
\fi

\ifdefined \poly 

\else
  \newcommand{\poly}{\optfont{poly}}
\fi

\usepackage{stackengine}

\usepackage[normalem]{ulem}

%% file: Notations_Quantum22.tex
\DeclareMathOperator*{\tensor}{\otimes}

\providecommand{\tr}{tr}

\ifdefined \Tr 
  \renewcommand{\Tr}[1]{\tr \Big\{#1\Big\}}
\else
  \newcommand{\Tr}[1]{\tr \Big\{#1\Big\}}
\fi

 \def\id{I_d}

\def\Loss{\mathcal{L}}
\def\parameter{\overrightarrow{a}}

\def\qps{\sigma^\bfs}

\newcommand{\optfont}[1]{\mathsf{#1}}

%% file: acronyms.tex
\usepackage[nolist,nohyperlinks]{acronym}

\newacro{ptp}[PtP]{Point-to-Point}
\newacro{iid}[i.i.d.]{independent and identically distributed} 
\newacro{IID}[i.i.d.]{independent and identically distributed} 
\newacro{PAC}[PAC]{\textit{probably approximately correct}}
\newacro{VC}[VC]{Vapnik–Chervonenkis}
\newacro{ERM}[ERM]{\textit{empirical risk minimization}}
\newacro{SVM}[SVM]{support-vector machine}
\newacro{SGD}[SGD]{stochastic gradient descent}

\newacro{POVM}[POVM]{positive operator-valued measure}
\newacro{QPAC}[QPAC]{\textit{quantum probably approximately correct}}

\newacro{QSRM}[QSRM]{\textit{quantum shadow risk minimization}}
\newacro{QC}[QC]{quantum computer}
\acrodefplural{OC}[OC's]{quantum computers}
\newacro{ML}[ML]{machine learning}
\newacro{QML}[QML]{quantum machine learning}
\newacro{NISQ}[NISQ]{noisy intermediate-scale quantum}
\newacro{VQA}[VQA]{variational quantum algorithm}
\newacro{QNN}[QNN]{quantum neural network}
\newacro{PQC}[PQC]{parameterized quantum circuit}

\newacro{RQSGD}[RQSGD]{randomized quantum stochastic gradient descent}
\newacro{QSGD}[QSGD]{quantum shadow gradient descent}
\newacro{QGD}[QGD]{quantum gradient descent}
\newacro{QSS}[QSS]{quantum shadow sampling}
\newacro{CST}[CST]{classical shadow tomography}
\newacro{PSR}[PSR]{parameter shift rule}

%% file: intro.tex
\section{Introduction}

\Ac{QML} and optimization stand out as important applications of \acp{QC} in the context of classical and quantum data with a diverse range of \textit{near-term} applications, including quantum simulation \cite{Kassal2011,McArdle2020,Hempel2018,Cao2019,Bauer2020}, quantum sensing,  phase-of-matter detection \cite{Carrasquilla2017,Broecker2017}, ground-state search \cite{Carleo2017,Broughton2020,Biamonte2017}, and entanglement detection \cite{Ma2018,massoli2021leap,Lu2018,Hiesmayr2021,Chen2021a,Deng2017}. 
A significant bottleneck in the utility of QML  for classical data  is the process of converting classical data into a quantum form, suitable for processing in quantum systems, and the practical implementation of Quantum Random Access Memories (QRAM) \cite{Giovannetti2008}.   However, this bottleneck can be bypassed when the data is inherently quantum upon collection as in quantum sensing applications. This naturally quantized data can then be directly utilized in QML algorithms, avoiding the overhead associated with classical-quantum encoding. Using quantum states as input, QML can leverage unique properties of quantum to achieve a greater accuracy and efficiency compared to classical methods. However, to fully harness the potential of QML with quantum data, it is essential to design training methods that optimize the relevant model parameters. This involves developing strategies to maximize the accuracy of QML approaches by leveraging the unique properties of quantum data.


\Acp{VQA} play a critical role for training of QML models, such as  \acp{QNN}, using a hybrid quantum-classical loop. 
Gradient-based approaches are commonly used as optimizers in such models \citep{Mitarai2018,Schuld2018,farhi2014quantum}. Despite their success in classical problems, these methods pose significant challenges due to the unique properties of quantum systems. Specifically, estimation of the gradient is a demanding task due to the state collapse, the information loss of quantum measurements and measurement \textit{incompatibility}.  As such, existing methods such as \textit{Parameter-shift rule} and finite differencing \cite{Mitarai2018,Schuld2018} share a common limitation --- they are sample-intensive, requiring fresh samples for estimating each component of the gradient. The inefficiencies are  exacerbated when quantum data is expensive to produce. For instance when state preparation has high gate complexity or quantum sensing is time consuming \footnote{An explicit example is when when quantum samples are the output of quantum dynamical process described by a a two-dimensional Fermi-Hubbard model on an $8 \times 8$ lattice --- a process that requires $10^7$ Toffoli gates \cite{Kivlichan2020}.}.

Efficient gradient estimation is therefore  crucial for the practical implementation of QML algorithms in resource-constrained or cost-sensitive applications. In this paper, we propose a new method to efficiently estimate the gradient for QML applications with quantum data.  We provide a through analysis of the convergence rate and show that our approach outperforms existing methods. Our approach consumes a single  quantum sample and updates all parameters simultaneously. 

\subsection{Summary of the main contributions}
This paper proposes a sample-efficient gradient-based method called \ac{QSGD} for training in \acp{VQA} (Section \ref{subsec:QSGD}). \Ac{QSGD}  consumes only one sample per each iteration and simultaneously updates all the parameters. It has the benefits of a one-shot approach in being sample efficient, while having a faster convergence rate compared to existing works. 
We provide comprehensive theoretical analysis for all claims and verify our results through numerical experiments. 

Our approach introduces a new technique for gradient estimation via \ac{CST}, see Section \ref{subsec:shadow}. Shadows from quantum states has been previously studied in the context of quantum state tomography \cite{Aaronson2018,Huang2020}.  In our work, we  utilize  this concept to the broader context of optimization and learning, propose a gradient-based optimizer and analyze its convergence rate in generic VQA setups, including non-convex landscapes, see Theorem \ref{thm:convergence}. 

Particularly, we show that for generic loss functions and \textit{local} quantum systems,  QSGD exhibits faster convergence compared to the \ac{PSR} and RQSGD \cite{HeidariAAAI2022}. Conventionally, we analyze the norm of the gradient to measure the convergence rate. Consider training of an ansatz with $p$ parameters  using  gradient-based update rule and with fixed-step size. We show that the average of the norm of the gradient is bounded by $\epsilon + O( \frac{C}{n \epsilon})$, 
where $n$ is the sample size and $C$ is the method-dependent constant characterized as 
\begin{align*}
C_{QSGD} = p^{3/2} 3^k, & & C_{PSR} = 4 p^{5/2}, & &  C_{RQSGD} = p^{5/2}
\end{align*}
where  $k$ is the ansatz locality. 

 Notably, our results show that QSGD achieves accelerated convergence for a $k$-local ansatz with $k=\Omega(\log p)$, where $p$ is the number of parameters --- a scenario commonly encountered in highly parameterized local ans\"atz. The term $k$-locality refers to operations that act non-trivially on $k$ out of $d$ qubits, a well-explored concept in quantum literature. For instance, the $k$-local Hamiltonian problem falls within the QMA complexity class (the quantum analog of NP) and shares similarities with the MAX $k$-SAT problem \cite{Kempe2004,Kitaev2002}.

Finally, we present numerical experiments in Section \ref{sec:numerical}, which support our theoretical findings and showcase the benefits of QSGD.

%% file: model.tex
\section{Preliminaries and Model}

\noindent\textbf{Notation:}  For any $d\in \NN$, we denote 
the set $\set{1,2,...,d}$ by $[d]$. 
 A quantum measurement $\mathcal{M}$ is  a \ac{POVM} represented by a set of operators $\mathcal{M}:=\{M_v, v\in\mathcal{V}\}$, where $\mathcal{V}$ is the set of possible outcomes, $M_v\geq 0$ for any $v\in \mathcal{V}$, and $\sum_{v\in \mathcal{V}} M_v =\id.$ The Pauli operators, along with the identity operator are denoted by $\sigma^0:=I, \sigma^1:=X, \sigma^2:=Y,$ and $\sigma^3:=Z$. The $d$-fold tensor products of these operators, called Pauli strings, are denoted by $\qps:=\tensor_{j}\sigma^{s_j}$, where $\bfs\in \set{0,1,2,3}^d$.    

\subsection{Variational Quantum Algorithms}
\acp{VQA} consists of an ansatz which is  a \ac{PQC} followed by a hybrid quantum-classical loop for optimisation. The ansatz is denoted by a parameterized unitary  $U(\overrightarrow{a})$, with $\overrightarrow{a}\in \RR^p$ being the vector of adjustable model parameters, and $p$ the number of the parameters. An important example of an ansatz is a \ac{QNN}, in which $U$ consists of several layers of parameterized quantum perceptrones.  
The PQC consists of  multiple layers of variational elementary gates as 
\begin{align}\label{eq:ansatz}
  U(\overrightarrow{a}) = U_L(\overrightarrow{a}_L) V_L \cdots U_1(\overrightarrow{a}_{1})V_1,
  \end{align}
where $L$ is the number of layers, and  $V_l$ is a fixed part followed by a parametric part $U_l$ given by : 
\begin{equation}\label{eq:ansatz U_l}
    U_l(\overrightarrow{a}_l):= \bigotimes_{\bfs_l} e^{i a_{\bf s_l} \sigma^{\bfs_l}}
\end{equation}
 with $\sigma^{\bfs}$ being a Pauli string.  Hardware-efficient ans\"atze are examples of such \acp{PQC} 
The fixed layers often consist of CNOT gates.  The objective of VQA is to minimize a predefined loss function characterized via an \textit{observable} ${O}$: 
\begin{equation}\label{eq:VQA cost function}
     c(\parameter) : = \tr{{O} ~ U(\overrightarrow{a}) \rho  U(\overrightarrow{a})^\dagger},
\end{equation}
where $\rho$ is a mixed state representing the input to the ansatz. The above formulation reveals a significant challenge for VQAs -- the cost function $c(\parameter)$ is in the form of an expectation value, which can only be measured via the observable $O$.  However, the mere act of measuring the loss causes state collapse, generally rendering it useless for further processing. Consequently, one expects that quantum optimization and learning problems are more sample intensive compared to their classical analogs. This is particularly problematic when the quantum states are expensive to produce.  

\noindent\textbf{Quantum Learning.}
An imperative application of VQAs is the ability to infer from quantum states specially in the context of \ac{QML}. QML is defined for both originally classical or quantum data. In the first case, classical information is encoded into qubits \cite{Biamonte2017} via various encoding procedures. Typically the classical-quantum encoding is a major bottleneck from computational perspective.  In the second case, the data is originally quantum obtained trough various methods including quantum sensors and therefore there is no encoding complications. The ability to process quantum data opens the   possibility for advancing our understanding of the kinds of quantum advantages over classical models.  This possibility comes with unique challenges stem from various factors such as quantum state collapse and the no-cloning principle. Such features prohibit reuse of quantum samples and hence ask for  QML solutions that abide such properties. The focus of this paper is on addressing such concerns.    

We consider a  quantum supervised learning suited for quantum state discrimination. In that scenario, one aims to accurately predict a classical property (a.k.a., label) of an unknown quantum state by applying an optimized ansatz with a measurement at the end.  If $\hat{y}$ is the ansatz's prediction of the true label $y$, a loss value of $l(y, \hat{y}) >0$ will be incurred. The goal is to minimize the generalization loss, as a function of $\parameter$, which is the expectation $\Loss(\parameter):=\EE[l(Y, \hat{Y})]$ over all the randomness involved in the problem. To supervise the learner, a set of $n$ labeled  samples $(\ket{\phi_i}, y_i), i\in [n]$  is available for training the ansatz. Typically, it is assumed that the samples are generated \ac{iid} according to an unknown but fixed probability distribution $D$. The learning algorithm aims to minimize the average loss for the training samples.  However, even the training loss is challenging to compute as it is an expectation value.  This is due to Born's law for quantum measurements. To see this, let $\{M_{\hat{y}}: \hat{y}\in \mathcal{Y}\}$ represent the measurement at the end of the ansatz. Then, the loss for the fixed sample $(\ket{\phi_t}, y_t)$ is:
\begin{equation}\nonumber
\Loss( \parameter^{(t)}, \ket{\phi_t}, y_t)  =  \sum_{\hat{y}} \ell(y, \hat{y})  \expval{U^\dagger(\overrightarrow{a}) {M}_{\hat{y}}U(\overrightarrow{a})}{\phi_t}. 
\end{equation}
Moreover, the average training loss is $\frac{1}{n}\sum_{t=1}^n \Loss( \overrightarrow{a}, \ket{\phi_t}, y_t).$ While it is desirable to minimize expected loss, this goal is not feasible since the training loss is an expectation value. In this context, designing sample-efficient optimizers is essential to real-world applications of QML. 

\subsection{Gradient-based optimizers}\label{sec:gradient optimizers}
Making iterative progress in the direction of the steepest descent is one of the most popular optimization techniques in \acp{VQA}, as it has been in classical problems.  
Ideally, if the per-sample expected loss $\Loss(\overrightarrow{a}, \ket{\phi_t}, y_t)$ were known, one would apply a standard gradient descent method  via an update rule of the form:
\begin{align}\label{eq:QSGD ideal}
\overrightarrow{a}^{(t+1)} = \overrightarrow{a}^{(t)} - \eta_t \nabla \Loss(\overrightarrow{a},\ket{\phi_t}, y_t).
\end{align}
This update rule is infeasible in practice, as the exact value of the gradient is unknown. This is because $\Loss(\overrightarrow{a},\rho_t, y_t)$ is an expectation value, and also, the superposition coefficients of $\ket{\phi_t}$ are unknown. Computing these coefficients requires infinitely many measurements, which is infeasible;  hence, one can only estimate/ approximate the gradient. 

\noindent\textbf{Gradient Estimation.}
There have been several approaches to estimate the gradient \cite{Farhi2018,Mitarai2018,Sweke2020,HeidariAAAI2022,Harrow2021,Schuld2018,Mitarai2019,Wiersema2023}. The zeroth-order approach (e.g., finite differences) evaluates the objective function in the neighborhood of the parameters.  Although it is a generic approach, recent studies showed their drawbacks in terms of convergence rate \cite{Harrow2021}.  First-order methods (e.g., parameter shift rule)  directly calculate the partial derivatives \cite{Schuld2018}. When the ansatz is of the form $U(\parameter) = \prod_{j=1}^p e^{-i a_j G_j}$, with $G_j$ the parameter shift rule implies that: 
\begin{equation*}
\pdv{\Loss}{a_j} = \<\Loss(\parameter + \frac{\pi}{4}e_j)\> -\<\Loss(\parameter - \frac{\pi}{4}e_j)\>, 
\end{equation*}
where $e_j\in \RR^p$ is the $j$th canonical vector, that is $e_{j,j}=1$ and $e_{j,r}=0$ for all $r\neq j$.  It is shown that this expression can be directly measured via a Hadamard test \cite{Mitarai2018}. In this way, each partial derivative is estimated by measuring $m$ fresh samples, where $m$ is a hyper parameter to be optimized. In this case, with $n$ total samples, the parameter shift rule will run for $T=\frac{n}{mp}$ iterations. Through standard statistical analysis, one can show that $\order{\frac{p}{\epsilon^2}\log p}$ samples are needed to estimate the gradient up to an additive error $\epsilon$.

A recently proposed alternative method, RQSGD updates a single parameter in each iteration \cite{HeidariAAAI2022}. For this, a single sample is measured via a Hadamard test to obtain a one-shot stochastic estimate of the selected partial derivative. In other words, the update rule is of the form:
$a_j^{(t+1)} = a_j^{(t)} - \eta_t  g^{(t)}_j,$
where $a_j$ is the selected parameter at iteration $t$, and $g^{(t)}_j$ is the estimate with $\EE[g^{(t)}_j] = \pdv{\Loss}{a_j}$. This update rule can be viewed as a quantum analog of the gradient coordinate descent \cite{Nesterov2012}. RQSGD offers high sample efficiency and low overhead. With $n$ samples, it runs for $T=n$ iterations, hence more updates of the ansatz. While this method is highly sample efficient, it  suffers from a slow convergence due to the serialization of updates. Hence, a key question is whether the sample efficiency can be maintained without compromising the convergence rate.

\noindent\textbf{No-cloning and measurement incompatibility.} 
Measurement incompatibility refers to the fact that certain properties of a quantum system cannot be simultaneously measured with arbitrary precision \citep{MichaelA.Nielsen2010}. This is a hurdle in the gradient estimation, as the gradient consists of several partial derivatives that might not be simultaneously measurable. As a result, existing methods  need fresh samples for estimating each partial derivative. This contributes to a  low sample efficiency that is challenges training of QML models with large number of parameters. This estimation deficiency is further exacerbated for quantum data,  as new quantum samples will be needed in view of the no-cloning.     We address this important issue using shadow tomography that enables high sample efficiency while having a faster convergence rate.

%% file: ShadowGradient.tex
\section{Main Results}


This section introduces our method, \ac{QSGD}, and provides theoretical underpinnings for its correctness and performance. First, in Section \ref{subsec:shadow gradient}, we study the gradient of the loss for product  ansatz. 
Next, in Section \ref{subsec:shadow},  we present our QSGD training technique and describe our estimation method. Lastly, in Section \ref{subsec:convergence}, we derive convergence results (Theorem \ref{thm:convergence}). 



\subsection{Gradient of the loss}\label{subsec:shadow gradient}

We start by taking the derivative of the loss with respect to the parameters of the ansatz given in \eqref{eq:ansatz}.  
For any $l\in [L]$, let $U_{\leq l}: = \prod_{r=1}^l U_r(a_r) V_r$ denote effective unitary till layer $l$. Similarly, let $U_{>l}:=\prod_{r>l} U_r(a_r)V_r$ denote the products after $l$. The derivative of the loss with respect to $a_l$ is then computed as follows:
    \begin{lemma}\label{lem:loss derivative}
Let ${\rho}_l^{out} =U_{\leq l}\ketbra{\phi} U_{\leq l}^\dagger$ denote the density operator of the output state at layer $l$ when the input is $\ket{\phi}$ with label $y$. Then, the derivative of the loss is given by: 
\begin{align}\label{eq:losss derivative}
\frac{\partial  \Loss
}{\partial a_{\bfs_l}} = \sum_{\hat{y}} i \ell(y,\hat{y})  \tr{M_{\hat{y}} U_{>l} \big[ \sigma^{\bfs_l}, {\rho}_l^{out}\big]U^\dagger_{>l}},
\end{align}
where $[\cdot, \cdot]$ is the commutator operation.
\end{lemma}
This is a known result for completeness the proof is provided in Appendix \ref{proof:lem:loss derivative}.
 The expression in the lemma can be implemented using a quantum circuit with measurements at the end. The circuit is shown in Fig.~\ref{fig:loss derivative circuit multi} which is a special example of the generalized Hadamard test \cite{Mitarai2019,Harrow2021}.  In this procedure, given a sample $\ket{\phi}$ and an ancilla qubit $\ket{0}$, we first apply the first $l$ layers of the ansatz, and then apply a special circuit with controlled rotations for taking the derivative of the loss. Then, the rest of the layers of the ansatz are applied and measurement is performed. Given the outcome $(\hat{y},z)$, we compute $g_l :=-2(-1)^z \ell(y, \hat{y})$ as the measured derivative, where $\ell(\cdot, \cdot)$ is the target loss function. We have the following result with the proof given in Appendix \ref{proof:lem:gradient unbiased}.
\begin{lemma}\label{lem:gradient unbiased}
Let $(\hat{y}, z)$ be the output of the circuit in Fig. \ref{fig:loss derivative circuit multi} and  $g_l :=-2(-1)^z \ell(y, \hat{y})$. Then, $\EE[g_l] = \frac{\partial  \Loss( \parameter, \ket{\phi}, y)
}{\partial a_{l}}$.
\end{lemma}

\begin{figure}
    \centering
    \resizebox{0.6\textwidth}{!}{
\begin{quantikz}[scale=0.3]
 \lstick{$\ket{\phi}$}&\gate{U_{\leq l}} &\gate{R_{\bfs_l}(+\frac{\pi}{2})}&\qw  &\gate{R_{\bfs_l}(-\frac{\pi}{2})} &\qw & \gate{U_{>l}} & \meter{} \arrow[r]&\rstick{$\hat{y}$}\\
 \lstick{$\ket{0}$} &\gate{H} &\ctrl{-1}  &\gate{X}& \ctrl{-1}  & \gate{X}& \qw & \meter{} \arrow[r]&\rstick{z}
\end{quantikz}
}
    \caption{Circuit for measuring the partial derivative with respect to a parameter $a_{\bfs_l}$ appearing at layer $l$. Here $U_{\leq l}$ corresponds to the first $l$ layers of the ansatz, and $U_{>l}$ to the remaining layers. Here, $X$ is the X-gate and $R_{\bfs_l}$ is the controlled rotation around Pauli $\sigma^{\bfs_l}$. 
    }
    \label{fig:loss derivative circuit multi}
\end{figure}
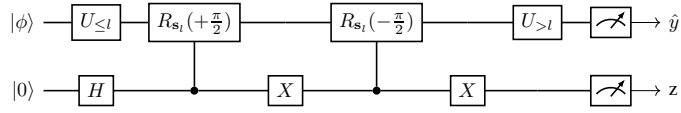

\subsection{Gradient Estimation}\label{subsec:shadow}
\nocite{Heidari2023Quantum1}

Our approach  is inspired by \ac{CST}  \citep{Aaronson2018,Huang2020}, where an approximate description of an unknown
quantum state is obtained by repeated measurements on its exact copies. CST is an approach to  estimate several observables of a quantum state with minimum copy complexity.  Classical shadows are matrix descriptions of the state, and are derived from measurements on randomly rotated bases. The Shadows are used to compute the expectation of any given observable $\CM$ via classical computations. In whis work we utilize CST to efficiently train VQAs. We show that with the shadows one obtains an unbiased estimate of the gradient that leads to convergence to a local minima.    

\noindent\textbf{Classical Shadow Tomography.} We start by describing the procedure for \ac{CST} using Pauli measurements. Let $\rho$ be a $d$-qubit  state (density operator). Then $d$ unitary operators $V_1,\cdots, V_j$ are selected randomly from the  set $\set{ I, H, S^\dagger H}$ with equal probabilities,  where $H$ is the Hadamard gate and $S=\sqrt{Z}$ is the square root of the Pauli-$Z$. The gates are applied to the input state followed by measurement in the computational basis.  The result is a binary string $\bfb^\rho \in \set{0,1}^d$ that is stored classically. Based on the choice of the unitaries and the measurement outcomes, the classical shadow of $\rho$ is a matrix computed as 
\begin{equation}\label{eq:rho hat shadow}
\hat{\rho}:=\bigotimes_{j=1}^d \qty(3 V^\dagger_j \ketbra{b^\rho_j}V_j - I).
\end{equation}

 Surprisingly, even though this  processes operates independently on each qubit, it creates an unbiased shadow even when the state is entangled across the $d$ qubits. In other words, $\EE[\hat{\rho}]=\rho$, where the expectation is taken over all the involved randomness including the selection of $V_j$ and the measurement outcomes $b^\rho_j$. Moreover, since the shadows are classical matrices, no-cloning will not be an issue. Predictions with CST has been studied originally using the median of means estimation \cite{Huang2020}.  In this work, we use the empirical average as the estimator of the gradient components. CST with Pauli measurements is advantageous when the  observables are local. There are several applications of \textit{near-term} \acp{QC} with local observables in many-body quantum physics \cite{Huang2020}, quantum chemistry \cite{Kandala2017}, lattice gauge theory \cite{Kokail2019}, and learning and optimization \cite{Chen2022,Atici2007,Saglam2018}.   In what follows, we describe the estimation process for $k$-local observables. 
 
 An observable $\CM$ on $d$ qubits is called $k$-local if it acts on at most $k <d$ qubits non trivially.  In other words, its operator can be written as $I_{d-k}\tensor M_{k},$ where $M_{k}$ is a non-trivial operator on the space of $k$ qubits, indexed by $q_1, \cdots, q_k$.  Typically, one takes a fixed $k$ much smaller than $d$. 
Given $n$ copies (samples) of a mixed state $\rho$, the objective is to estimate the expectation $\<\CM\>_{\rho}$\footnote{This is equivalent to the case where $n$ quantum samples $\ket{\phi_i}, i\in [n]$ are drawn from a distribution $D$ so that  $\rho=\EE_D[\ketbra{\phi}]$.}.  For this, we apply CST on the $i$th copy of $\rho$ resulting in the bit string $\bfb^\rho_i$ and the index of the chosen unitary matrices $V_{i, 1}, \cdots, V_{i, d}$. Then, we compute a $k$-local shadow using \eqref{eq:rho hat shadow} but on the relevant qubits: 
\begin{align}\label{eq:stochastic QSS}
\hat{\rho}^k_{i} := \bigotimes_{j=1}^k \qty(3 V^\dagger_{i, q_j} \ketbra{b^\rho_{i, q_j}}V_{i, q_j} - I).
\end{align}
 Lastly, we apply the following estimator:
\begin{align}\label{eq:CM hat}
  \hat{S}:=\frac{1}{n}\sum_{i\in [n]}\tr{\CM \hat{\rho}^k_{i}},
\end{align}
This  procedure can be easily extended to estimate multiple observables. In what follows,  we show that this estimator is unbiased with the proof given in Appendix \ref{proof:thm:shadow unbiased}.

\begin{theorem}\label{thm:shadow unbiased}
For the input state $\rho$ and the $k$-local observable $\CM$, the estimation $\hat{S}$ in \eqref{eq:CM hat} is unbiased, that is  $\EE[ \hat{S}]=\<\CM\>_{\rho}$, where the expectation is taken over the sample distribution and the randomness in QSS. Moreover, the square error of the estimation is
$\var(\hat{S}) \leq \frac{3^k}{n} \norm{\CM}^2_\infty.$
\end{theorem}
The bound on variance of the estimation can be improved in terms of the factor $9^k$ and the infinity norm. 
It is worth noting that the drawback of CST is when the observables are not local leading to a large error term. In that case, one needs to use Clifford gates to reduce the estimation error. However, this process might not be computationally tractable due to exponential dimensions. 

\subsection{Quantum shadow gradient descent (QSGD)}\label{subsec:QSGD}
We use CST to measure the gradient of the loss by applying it to each input sample. With this procedure and Lemma \ref{lem:loss derivative}  we can estimate all the derivatives of the loss. 

At each iteration $t$, given the input sample $(\ket{\phi_t}, y_t)$, we generate the corresponding shadow  $\hat{\rho}_t$ and pass it through the ansatz. Let $\hat{\rho}_l^{out}$ be the output of the ansatz at layer $l$ when the input is the shadow $\hat{\Phi}_t$. Then, we obtain the derivatives of the gradient via the circuit in Fig. \ref{fig:loss derivative circuit multi}. 
Continuing this for all the layers with copies of $\hat{\Phi}_t$, we obtain the shadow gradient for the current sample denoted by $ \widehat{\nabla \Loss} ( \parameter, \hat{\rho}^{out}_t, y)$.  With this procedure, we apply the following update rule: 
\begin{align}\label{eq:QSGD shadow}
\parameter^{(t+1)} = \parameter^{(t)} - \eta_t \widehat{\nabla \Loss}(\parameter,\hat{\rho}^{out}_t, y_t).
\end{align}

From this, we obtain our QSGD process, which is summarized in Alg.~\ref{alg:QSGD}.

\begin{algorithm}[ht]
\caption{QSGD}\label{alg:QSGD}
\begin{algorithmic}[1]
\STATE Initialize $\parameter^{(1)}$ randomly.
\FOR{$t=1$ to $n$}
\STATE Compute $\hat{\rho}_t$ the shadow of $\ket{\phi_t}$ as in \eqref{eq:rho hat shadow}. 
\STATE Estimate the partial derivative for each parameter $a_{\bfs_l}$ with CST and with the input being  $\hat{\rho}_t$. 
\STATE Apply the update rule 
$\parameter^{(t+1)} = \parameter^{(t)} - \eta_t \widehat{\nabla \Loss}.$
\ENDFOR
\STATE \textbf{return} $\overrightarrow{a}^{(n+1)}.$
\end{algorithmic}
\label{ald:osgd}
\end{algorithm}

\subsection{Convergence analysis}\label{subsec:convergence}
Research in the development of convergence guarantees for stochastic gradient descent  with general loss functions is an active area of research even in classical learning and optimization. This endeavor holds significance in establishing a robust theoretical foundation for machine learning algorithms \cite{Bottou2018,Wolf2023}. Convergence guarantees of SGD have been studied under special classes of loss functions under various conditions such as convex, strong-convex, and Polynak-\L ojasiewicz (PL) inequality (for a detailed discussion see \cite{Wolf2023} and the references therein). Similar guarantees have been presented for \acp{VQA} under various constraints \cite{Sweke2020,Harrow2021}.  However, \acp{VQA} frequently feature objective functions that are non-convex and deviate from these conditions. 

In this paper, we present a theoretical analysis for generic non-convex loss, offering insights into the comparison of various methods. We compare the convergence of three approaches based on the number of quantum samples for (i) the parameter shift rule; (ii) the one-shot  RQSGD; and (iii) our proposed QSGD. 
\begin{theorem}\label{thm:convergence}
    Consider a bounded (potentially non-convex) loss function and let ${l}_{\max}:=\sup_{y, \hat{y}}\ell(y, \hat{y})$. If we apply QSGD with fixed step size $\eta = \frac{\epsilon}{C}$, with  $\epsilon\in (0,1)$, to train an ansatz with $p$ number of parameters and with $n$ quantum labeled samples, then there will be $T=n$ iterations and the gradient decays as
    \begin{align*}
    \EE\qty[\frac{1}{T}\sum_{t=1}^T \norm{\nabla\Loss(\parameter^{(t)})}_2^2] \leq \epsilon +C \frac{ l_{\max}^3 }{n \epsilon}\qty(\Loss(\parameter^{(1)})-\Loss(\parameter^{*})),
\end{align*}
where  $C=2 p^{3/2} 3^k$ with $k$ being the locality parameter, and $\parameter^{*}$ is the global best parameter. Moreover, for RQSGD $T=n$ and  $C = p^{5/2}$; whereas for the parameter-shift rule  $T= \frac{n}{mp}$ with $m$ being the number of measurement shots and  $C = 4p^{5/2}$.
\end{theorem}
Details of the convergence analysis and the proof of the theorem is given in Appendix \ref{sec:proof conv analysis}. We also  present the convergence guarantees to global minimum for strongly convex loss, see Appendix \ref{app:strongly convex} for the details.  


This result implies that RQSGD and the parameter shift rule have a similar convergence rate that is $\epsilon+O(\frac{p^{5/2}}{n \epsilon})$; whereas QSGD has a lower gap of $\epsilon+O(\frac{p^{3/2}9^k}{n\epsilon})$, at least when the ansatz's locality is $k =O(\log p)$. 
Our numerical results in the next section verify this finding.

%% file: Numerical.tex
 \section{Numerical Experiments}\label{sec:numerical}
\begin{figure}
    \centering
    \includegraphics[scale=0.4]{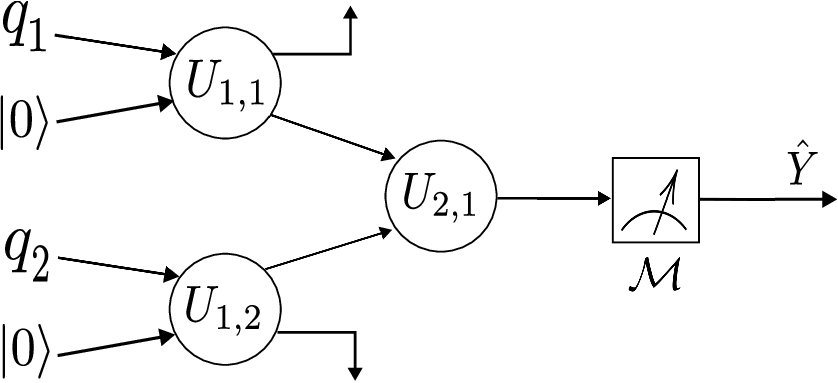}
    \caption{The QNN setup for Exp. 1 and 2.}\label{fig:QNN setup}
\end{figure}
We present experimental results in support of our theoretical claims of accuracy and performance. In each of the comparisons, we use three gradient-based methods: (i) Parameter Shift Rule; (ii) RQSGD; and (iii) the proposed QSGD.  Each experiment consists of an ansatz with a quantum dataset for the task of classification and training. The setup of each experiment is described below.

\input{figures_comparison.tex}

\noindent\textbf{Experiment 1.}
The dataset is synthetically created from recent results in quantum photonics \cite{Mohseni2004,Patterson2021,Li_Song_Wang_2021}. Here, each quantum state is  a superposition of two qubits, with a label $y\in \set{-1, 1}$. Particularly,  each labeled sample is of one of the three forms: $(\ket{\phi_u}, -1), (\ket{\phi_{+v}}, 1), (\ket{\phi_{-v}}, 1)$ with equal probability, where 
\begin{align*}
\ket{\phi_u} := \sqrt{1-u^2}\ket{00}+u\ket{10},\qquad  
 \ket{\phi_{\pm v}} := \pm \sqrt{1-v^2}\ket{01}+v\ket{10},
\end{align*}
where each time $(u, v)$ is selected randomly and uniformly from $[0,1]^2$.

\noindent\textbf{Experiment 2.}
This dataset is procedurally generated to classify between \textit{separable} and \textit{maximally entangled} states. To do this, we randomly generate several $2$-qubit density operators of the two different types. Separable states are of the form $\rho_A\tensor \rho_B$, where $\rho_A$ and $\rho_B$ each are single-qubit density operators ($2\times 2$ matrices) generated randomly based on a Haar measure using \textit{RandomDensityMatrix} in \citep{qetlab}. Each maximally entangled state is of the form $$\ket{\psi}:=\frac{1}{4} (\ket{00}+\ket{01}+\ket{10}+\ket{11}).$$ For the labeling, we assign $y=-1$ to separable states and $y=1$ to maximally entangled states. 

\noindent\textbf{Experiment 3.}
 For testing the performance of our approach in higher dimension ($\dim = 2^4$),  we classify multi-qubits Greenberger-Horne-Zeilinger (GHZ) states (with label $y=1$) from the separable states (with label $y=-1$). The GHZ state for multi-qubits is a generalization of the three-qubit GHZ state. It is a type of entangled quantum state that involves multiple qubits and exhibits non-local correlations. We consider $4$-qubit GHZ state given by
$ \ket{GHZ} = \frac{\ket{0000} +\ket{1111}}{\sqrt{2}}.$
The separable states are generated in the same manner as Experiment 2. 
\begin{table*}[h!]
\centering
\begin{tabular}{l|ccccc}
\toprule
Dataset & QSGD  & RQSGD &  Parameter Shift Rule &Up. bound    \\
\midrule
Exp. 1 
                & 92.47\%     & 91.83 \%     & 86.88\% & 93.81 \%  \\ \hline
Exp. 2 
                & 94.18 \%     & 92.11 \%    & 89.88 \%  & 95.16\%\\ \hline
Exp. 3 
               & 98.29 \%       & 97.59 \%     &83.12 \%  & 99.78 \%  \\
\bottomrule
\end{tabular}
\caption{ Validation accuracy of the trained ansatz for all three experiments.  
}
\label{tab:validation}
\vspace{-10pt}
\end{table*}

\noindent\textbf{Ans\"atze:}
For Experiments 1 and 2, we use a QNN as the ansatz (see Fig. \ref{fig:QNN setup}). This QNN is composed of three perceptrons, 
where $U_{11}, U_{12}, U_{21}$ are parameterized unitary operators of the form 
$\prod_\bfs \exp\{i b_{\bfs} \qps\},$
and  $\bfs\in \set{0,1,2,3}^2$. Hence, there are $p=48$ parameters to be trained. 
The inputs are $2$-qubit states with two ancilla qubits $\ket{00}$.  The measurement performed on the readout qubits is the POVM $\mathcal{M}:= \{ \Pi_{-1}, \Pi_1\}$ with outcomes in $\{-1, 1\}$ and operators
$$\Pi_{-1}= I_2-\Pi_1 =\ketbra{00} + \ketbra{11},$$
where $I_2$ is the $2$-qubit identity.

 For the GHZ dataset, the ansatz is a unitary operator parameterized  as
 $\prod_\bfs \exp\{i b_{\bfs} \qps\}$ as before.
Hence, the unitary includes $p=4^4=256$ parameters for training.  
The output state is measured as POVM $\mathcal{M}':= \{ \Pi'_{-1}, \Pi'_1\}$,  with outcomes in $\{-1,1\}$ and operators defined as:
\begin{align*}
\Pi'_{-1}= I_4-\Pi'_1 = \ketbra{0000} + \ketbra{1111},
\end{align*}
where $I_4$ is the $4$-qubit identity.

\noindent\textbf{Loss.} We use the 0-1 loss functions defined as $l(y, \hat{y})=1$ when $y\neq \hat{y}$ and zero otherwise.  We use the prediction accuracy which is $1-l(y, \hat{y}).$

\subsection{Experimental Results}
We train the ansatz using three different methods: (i) our proposed QSGD, (ii) RQSGD, and (iii) Parameter Shift Rule. We compare these methods with a theoretical upper limit (see Lemma \ref{lem:optimal loss} in Appendix \ref{app:optimal loss}). 

Fig.~\ref{fig:comp} presents the training accuracy for the endorsed ansatz configuration, as a function of $n$, the number of samples across the three datasets.  
The figure supports the following claims: the convergence of QSGD is significantly faster than RQSGD and parameter shift rule, and is close to the theoretical upper limit. As expected, the gap is higher for the third experiment as the number of parameters is larger.  This supports our claims about the optimality gap. 

We also generated a validation set of 2000 samples to test and compare the training procedures' accuracy. These results are summarized in Table \ref{tab:validation}.  
 From this table, we note that the accuracy of QSGD is significantly better than the other methods and is close to the theoretical upper limit as well. 
 All experiments are simulated on a classical computer. The implementation details are provided in Appendix \ref{app:imp details} and source code is given in the Supplementary material.

%% file: figures_comparison.tex
  \begin{figure*}[ht!]
  \captionsetup[subfigure]{justification=centering}
  \centering
  \begin{subfigure}{0.32\textwidth}
        \includegraphics[width=\textwidth]{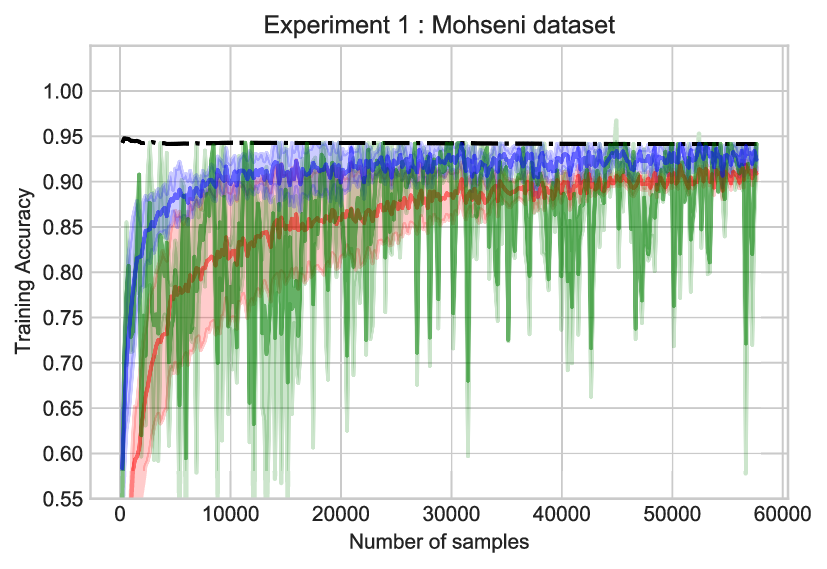}
        
\end{subfigure}
\begin{subfigure}{0.32\textwidth}
        \includegraphics[width=\textwidth]{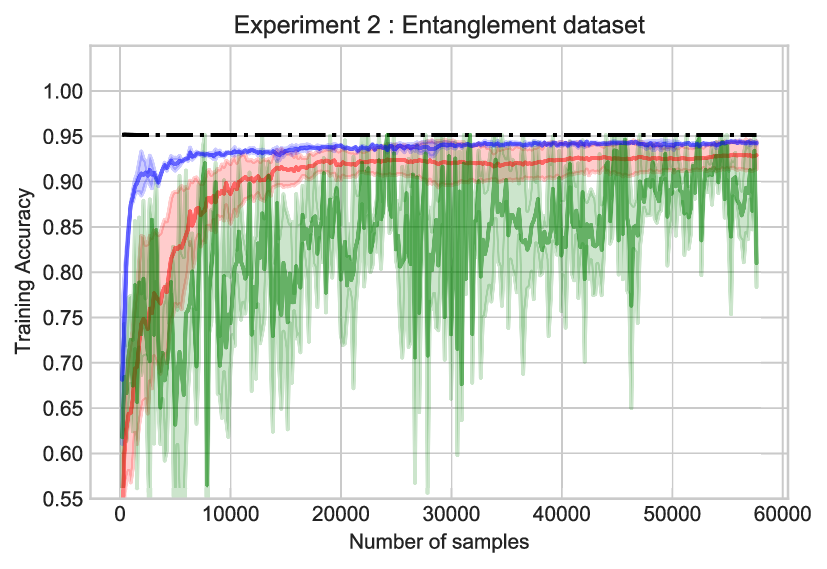}
\end{subfigure}
\begin{subfigure}{0.32\textwidth}
        \includegraphics[width=\textwidth]{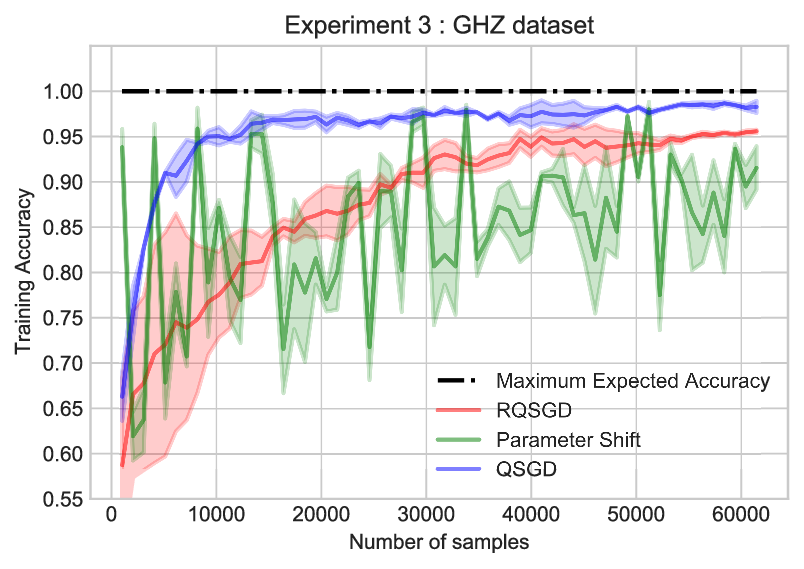}
\end{subfigure}

 \caption{Comparing the training accuracy in three experiments based on four methods: QSGD (this work), RQSGD, and the parameter shift rule gradient computation. The vertical axis is the training accuracy and the horizontal axis is the training samples. The cyan dashed line is the theoretical upper limit on the accuracy. All results were obtained using an optimized learning rate for each method by adopting a grid search technique in different (three) experimental settings.}
    \label{fig:comp}
\end{figure*}

%% file: RelatedWorks.tex
\section{Related results}\label{sec:related}
There has been significant recent interest in the development of VQAs and QNNs -- we highlight closely related results here.
VQAs, in general, have been studied for a variety of problems involving quantum or classical data \cite{cerezo2020variational,guerreschi2017practical}. Learning from quantum data has been studied extensively in recent literature in the context of diverse applications, including phase-of-matter detection \cite{Carrasquilla2017,Broecker2017}, ground-state search \cite{Carleo2017,Broughton2020,Biamonte2017}, entanglement detection \cite{Ma2018,massoli2021leap,Lu2018,Hiesmayr2021,Chen2021a,Deng2017}, and related problems in theoretical chemistry \cite{Kassal2011,McArdle2020,Hempel2018,Cao2019,Bauer2020}. VQAs are also applicable to classical data by mapping them to input quantum states~\citep{Giovannetti2008,Park2019,Rebentrost2014,Lloyd2013,Lloyd2014}.
 \acp{QNN} have received significant research attention since the early work \cite{Toth_1996,lewenstein1994quantum} to more recent developments over the past two  decades~\cite{Schuld_2014,Mitarai2018,Farhi2018,Torrontegui2018,Schuld_2020,Beer2020}.   Massoli et al. provide an excellent survey of these efforts~\citep{massoli2021leap}. VQAs have also been applied to the specific context of \acp{QNN}  ~\cite{Mitarai2018,Farhi2018,Schuld_2020,Beer2020,HeidariAAAI2022,cerezo2020variational,guerreschi2017practical}. 

%% file: proof_lem_loss_derivative.tex
\section{Proof of Lemma \ref{lem:loss derivative}}\label{proof:lem:loss derivative}
\begin{customlemma}{\ref{lem:loss derivative}}
Let ${\rho}_l^{out} =U_{\leq l}\rho U_{\leq l}^\dagger$ denote the density operator of the output state at layer $l$ when the input is $\rho$. Then, the derivative of the loss is given by: 
\begin{align*}
\frac{\partial  \Loss( a_{\bfs}, \ket{\phi}, y)}{\partial a_{\bfs_l}} = \sum_{\hat{y}} i \ell(y,\hat{y})  \tr{M_{\hat{y}} U_{>l} \big[ \sigma^{\bfs_l}, {\rho}_l^{out}\big]U^\dagger_{>l}},
\end{align*}
where $[\cdot, \cdot]$ is the commutator operation and $y$ is the label of $\rho$.
\end{customlemma}

\begin{proof}
The linearity of the derivatives yields the following:
 \begin{align*}
     &\frac{\partial  \Loss( \overrightarrow{a}, \ket{\phi_t}, y_t)}{\partial a_{\bfs_l}} =\sum_{\hat{y}} \ell(y_t,\hat{y}) \tr\Big\{ M_{\hat{y}}\frac{\partial }{\partial a_{\bfs_l}} \Big(U(\overrightarrow{a})\ketbra{\phi_t} U^\dagger(\overrightarrow{a})\Big)\Big\}.
 \end{align*}
Note that $U(\parameter) = U_{>l}U_{\leq l}$. Hence, the derivative only acts on the first $l$ layers. For that, $\frac{\partial }{\partial a_{\bfs_l}} U = iU_{>l}\sigma^{\bfs_l} U_{\leq l}$. Hence, the derivative inside the trace is given by:
\begin{align*}
&i \Big( U_{>l} \sigma^{\bfs_l} U_{\leq l} \ketbra{\phi_t} U^\dagger - U\ketbra{\phi_t} U_{\leq l}^\dagger \sigma^{\bfs_l} U_{>l}^\dagger\Big).
\end{align*}
Denote ${\rho}_l^{out} =U_{\leq l}\rho U_{\leq l}^\dagger$ as the output state. Then, the above term can be written as the commutator $U_{>l} \big[ \sigma^{\bfs_l}, {\rho}_l^{out}\big]U_{>l}^\dagger$. With this, the expression in the lemma is obtained. 
\end{proof}

%% file: Proof_thm_shadow_unbiased.tex
\section{Proof of Theorem \ref{thm:shadow unbiased}}\label{proof:thm:shadow unbiased}
\begin{customthm}{\ref{thm:shadow unbiased}}
For the input state $\Phi$ and the $k$-local observable $\CM$, the estimation $\hat{S}$ in \eqref{eq:CM hat} is unbiased, that is  $\EE[ \hat{S}]=\<\CM\>_{\Phi}$, where the expectation is taken over the samples distribution and the randomness in QSS. Moreover the square error of the estimation is
\begin{align*}
\var(\hat{S}) \leq \frac{3^k}{n} \norm{\CM}^2_\infty.
\end{align*}
\end{customthm}
\begin{proof}
For the first statement of the theorem, we start with taking the expectation: 
\begin{align*}
 \EE[{\hat{S}_\ell}]&= \frac{1}{nm}\sum_{i\in [n]}\sum_{r\in [m]}  \EE\Big[X_{ i,r} \prod_{j=1}^k W_{i,q_j}\Big]\\
 &=  \frac{1}{nm}\sum_{i\in [n]}\sum_{r\in [m]}  \EE\Big[ \tr{\CM \hat{\Phi}_i} \prod_{j=1}^k W_{i,q_j}\Big]\\
&=  \frac{1}{n}\sum_{i\in [n]}   \tr{\CM \EE\Big[\prod_{j=1}^k W_{i,q_j}\hat{\Phi}_i \Big]}
\end{align*}
where we used the linearity of the trace and expectation. Let $\mathcal{Q}=\{q_1, \cdots, q_k\}$, be the set of qubits that the observable $\CM$ acts on. From the definition of $\hat{\Phi}_i$ in \eqref{eq:stochastic QSS} we have that 
\begin{align}\label{eq:proof tensor state}
\prod_{j=1}^k W_{i,q_j}\hat{\Phi}_i = \bigotimes_{j\in  \mathcal{Q}} \big( W_{i,j} \ketbra{\omega_{j}^\Phi}\big) \bigotimes_{j'\notin \mathcal{Q}} \ketbra{\omega_{j'}^\Phi}
\end{align}

Let $\Gamma_0: \mathcal{B}[\mathcal{H}_2] \rightarrow \mathcal{B}[\mathcal{H}_2]$ be a linear mapping on $\mathcal{B}[\mathcal{H}_2]$, the space of all operators on the two dimensional Hilbert space $\mathcal{H}_2$. For any operator $B\in \mathcal{B}[\mathcal{H}_2]$, define: 
\begin{align*}
 \Gamma_0[B]:=\hspace{-10pt}\sum_{V\in \set{ I, H, S^\dagger H}}\sum_{b\in \set{0,1}} \frac{1}{3}\matrixelement{b}{V^{\dagger} B V}{b}~ V\ketbra{b}V^\dagger.
 \end{align*} 
 The inverse of this mapping has been used in single qubit Shadow tomography \cite{Huang2020}.  By direct calculation, it is not difficult to show that for a generic state $\ket{\bfa}=a_0 \ket{0}+a_1\ket{1}$, we have that: 
\begin{align*}
\Gamma_0^{-1}[\ketbra{\bfa}] = \begin{bmatrix}
2|a_0|^2-|a_1|^2 & 3a_0a_1^*\\
3a_0^*a_1 & 2|a_1|^2-|a_0|^2
\end{bmatrix}.
\end{align*}

Note that each $\ket{\omega^\Psi_j}$ in \eqref{eq:stochastic QSS} belongs to the set $\set{\ket{0}, \ket{1}, \ket{+}, \ket{-}, \ket{+i}, \ket{-i}}$. For any $c\in \set{0,1,+,-,+i, -i}$, let $\bar{c}$ denote its basis complement. That is, if $c=0$ then $\bar{c}=1$; if  $c=+$, then $\bar{c}=-$, and if $c=+i$, then $\bar{c}=-i$.   Hence, by direct calculation, we can show that applying $\Gamma_0^{-1}$ on $\ket{\omega^\Phi_j}$ gives:
\begin{align*}
&\Gamma_0^{-1}[\ketbra{\omega^\Phi_j}]=2\ketbra{\omega^\Phi_j}-\ketbra{\bar{\omega}^\Phi_j}.
\end{align*}
Note that by normalizing the above expression, we obtain the mixed state: $\frac{2}{3}\ketbra{\omega^\Phi_j}-\frac{1}{3}\ketbra{\bar{\omega}^\Phi_j}$. Now fix a choice of $V_j$ and $b_j^{\Phi}$ in the QSS process. Let $Q[b_j^{\Phi}] = (c_j, w_j)$ be the output of the channel $Q$.  Taking the expectation of $W_j\ketbra{\omega^\Phi_j}$ over the randomness in the channel $Q$ gives 
\begin{align*}
\EE\Big[W_j\ketbra{\omega^\Phi_j} \Big| V_j, b_j^{\Phi}\Big] &= \EE\Big[W_jV_j \ketbra{Q[b_j^{\Phi}]}V_j^\dagger  \Big| V_j, b_j^{\Phi}\Big]\\
&=\frac{2}{3} \times 3\times V_j \ketbra{b_j^{\Phi}}V_j^\dagger - \frac{1}{3} \times 3\times  V_j \ketbra{\bar{b}_j^{\Phi}}V_j^\dagger\\
&  = 2\ketbra{\omega^\Phi_j}-\ketbra{\bar{\omega}^\Phi_j}\\
&=\Gamma_0^{-1}[\ketbra{\omega^\Phi_j}].
\end{align*}

Note that $\hat{S}_\ell$ only depends on the terms in $\mathcal{Q}$ as  $\CM_\ell$ only acts on $\mathcal{Q}$. Hence, without changing $\hat{S}_\ell$, we can replace the second tensor states in \eqref{eq:proof tensor state} with the following states 
 \begin{align*}
  \bigotimes_{j'\notin \mathcal{Q}} \Gamma_0^{-1}[\ketbra{\omega^\Phi_{j'}}].
  \end{align*} 
This leads to the following expression for $\hat{S}_\ell$:
\begin{align*}
\EE[{\hat{S}_\ell}] = \frac{1}{n}\sum_{i=1}^n   \tr{\CM_\ell \EE\Big[\bigotimes_{j\in  \mathcal{Q}} \big( W_{i,j} \ketbra{\omega_{j}^\Phi}\big) \bigotimes_{j'\notin \mathcal{Q}} \Gamma_0^{-1}[\ketbra{\omega^\Phi_{j'}}] \Big]}
\end{align*}
Next, with the law of iterative expectations, we first calculate the expectation over the randomness of the channel $Q$.  Consequently, as the channel $Q$ is applied independently to each $j$, then
\begin{align*}
\EE[\hat{S}] &= \frac{1}{n}\sum_{i=1}^n   \tr{\CM \EE\Big[\bigotimes_{j\in  \mathcal{Q}} \Gamma_0^{-1}[\ketbra{\omega^\Phi_j}] \bigotimes_{j'\notin \mathcal{Q}} \Gamma_0^{-1}[\ketbra{\omega^\Phi_{j'}}] \Big]} = \frac{1}{n}\sum_{i=1}^n   \tr{\CM \EE\Big[\bigotimes_{j=1}^d \Gamma_0^{-1}[\ketbra{\omega^\Phi_j}]\Big]}\\
&=\frac{1}{n}\sum_{i=1}^n   \tr{\CM \EE\big[\hat{\Phi}_i\big]}\\
&= \frac{1}{n}\sum_{i=1}^n   \tr{\CM \Phi } = \<\CM\>_{\Phi},
\end{align*}
where the remaining expectation is over the randomness of $V_j$ and $b_j$ and  we used  \cite[Lemma 6]{Heidari2023Quantum1} and the definition of $\hat{\Phi}_i$. 

The second statement of the theorem concerns with the variance of this estimation. For that, we have that
\begin{align*}
 \EE[(\hat{S})^2] &= \frac{1}{n^2} \sum_{i,i'\in [n]} \EE\big[{X}_i{X}_{i'}\prod_{j=1}^k W_{i,q_j} W_{i',q_j}\big]\\
 &=   \frac{1}{n}\EE\big[{X}^2_{1} \prod_{j=1}^k W^2_{i,q_j}\big] + \frac{n-1}{n}(\<\CM\>_{\Phi})^2\\
&=\frac{9^k}{n}\EE\big[{X}^2_{1}\big] + \frac{n-1}{n} (\<\CM\>_{\Phi})^2
\end{align*}
With the law of total expectation, the first term is written as 
\begin{align*}
\frac{9^k}{n}\EE_{\hat{\Phi}_1}\big[\tr{\CM^2   \hat{\Phi}_1}\big] = \frac{9^k}{n}\tr{\CM_\ell^2  \EE[ \hat{\Phi}_1]} \leq \frac{9^k}{n} \norm{\CM}^2_\infty.
\end{align*}
Therefore, 
\begin{align*}
\var(\hat{S}) &= \EE[(\hat{S})^2] -  \<\CM\>_{\Phi}^2 \leq \frac{9^k}{n}\norm{\CM}_\infty^2  - \frac{1}{n}\<\CM\>_{\Phi}^2\\
&\leq  \frac{9^k}{n}\norm{\CM}_\infty^2.
\end{align*}
\end{proof}

%% file: convergence.tex
\section{Proof of Convergence Analysis}\label{sec:proof conv analysis}
In our analysis, we consider a generic SGD update rule described by: 
\[\parameter^{(t+1)} = \parameter^{(t)} - \eta_t g_t,\]
where $g_t$ is an unbiased estimate of the gradient $\nabla\Loss(\parameter^{(t)}).$ Our derivation depends on the  Lipschitz constant of the loss function, as calculated in \cite[Theorem 2]{Schuld2018}. 

\begin{remark}\label{rem:Lipschitz}
The loss function $\Loss(\parameter)$ is Lipschitz continuous with the constant $\mathfrak{L}=2 \sqrt{p} {l}_{\max}$, where $p$ is the number of ansatz parameters and ${l}_{\max}:=\sup_{y, \hat{y}}\ell(y, \hat{y})$.
\end{remark}
 Owing to the fact that the loss value and the ansatz parameters are classical we can theoretical guarantees on the convergence to local minima.   We present with the following result using which we define our optimality gap.



\begin{theorem}[\cite{Bottou2018}]\label{thm:non convex SGD}
Consider an SGD method with the sub-gradient $g_t$ and fixed step size $\eta$, such that (i) $g_t$ is unbiased,  (ii) $\EE[\norm{g_t}_2^2]\leq \mathfrak{M} + \mathfrak{M}_G \norm{\nabla\Loss(\parameter^{(t)})}_2^2$ for some constants $\mathfrak{M}, \mathfrak{M}_G$, and (iii) $0< \eta \leq \frac{1}{\mathfrak{L} \mathfrak{M}_G}$. Then, 
\begin{align*}
    \EE\qty[\frac{1}{T}\sum_{t=1}^T \norm{\nabla\Loss(\parameter^{(t)})}_2^2] \leq \eta\mathfrak{L}\mathfrak{M} + \frac{2 \qty(\Loss(\parameter^{(1)})-\Loss(\parameter^{*}))}{T\eta}.
\end{align*}
\end{theorem}
With this theorem one can find the rate of decay of the gradient as a measure of convergence to local minima. Specifically, we can prove the following result by setting $\eta=\frac{\epsilon}{\mathfrak{L}\mathfrak{M}}$.
\begin{corollary}
    In the setting for Theorem \ref{thm:convergence}, the average norm of the gradient decays as
    $\epsilon +   \order{\frac{\mathfrak{L}\mathfrak{M}}{T\epsilon}}.$
\end{corollary}

$\bullet$ \textbf{Parameter Shift Rule:} For this strategy 
$ 
 \frac{\mathfrak{L}~\mathfrak{M}}{T} = \frac{4p^{5/2}}{n} {l}^3_{\max}.
$
The argument is as follows. As described in Section \ref{sec:gradient optimizers}, with $n$ samples this method runs for $T=\frac{n}{m p}$ iterations, where $m$ is the number of measurement shots per parameter per iteration. Moreover, the factor $\mathfrak{M}$ as in Theorem \ref{thm:convergence} equals to $\mathfrak{M} = \frac{4p ~{l}^2_{\max}}{m}$ because of the followign result.  
\begin{lemma}
    The sub-gradients $g_t$ in the parameter shift rule satisfy 
   \[\EE[\norm{g_t}_2^2] \leq \frac{4p ~{l}^2_{\max}}{m} + (1- \frac{1}{m}) \norm{\nabla\Loss(\parameter^{(t)})}_2^2.\]
\end{lemma}
\begin{proof}

Note  the sub-gradient $g_t$ is a vector of $g_{t,\bfs_l}$, for all $\bfs_l$ appearing in ansatz formulation as in \eqref{eq:ansatz} and \eqref{eq:ansatz U_l}. In the parameter shift rule, each component is estimated empirically as 
$
g_{t,\bfs_l} = \frac{1}{m} \sum_{j=1}^m (x_j^+-x_j^-),
$
where $x_j^+$ and $x_j^-$ are the loss measurement where $a_{\bfs_l}$ is shifted by $\pm \frac{\pi}{4}$. As the measurements are done independently,  $\var(g_{t,\bfs_l})=\frac{1}{m}\var(X^+-X^-)$, where $X^\pm$ is distributed as each $X_j^\pm$. Since the estimation is unbiased we have that 
\[\var(X^+-X^-)=\EE[(X^+-X^-)^2] - (\pdv{\Loss}{a_{\bfs_l}})^2.\]
Therefore, given that $-l_{\max} \leq X^{\pm} \leq l_{\max}$, we can bound the above term as $4l^2_{\max}  - (\pdv{\Loss}{a_{\bfs_l}})^2.$ This implies that,   $\var(g_{t,\bfs_l})\leq \frac{1}{m}(4l^2_{\max}  - (\pdv{\Loss}{a_{\bfs_l}})^2)$. Hence, 
\begin{align*}
\EE[\norm{g_t}_2^2] &= \sum_{\bfs_l} \EE[|g_{t,\bfs_l}|^2] = \sum_{\bfs_l}\var(g_{t,\bfs_l}) + (\pdv{\Loss}{a_{\bfs_l}})^2\\
&\leq \frac{4p ~{l}^2_{\max}}{m} + (1- \frac{1}{m})\sum_{\bfs_l} (\pdv{\Loss}{a_{\bfs_l}})^2.
\end{align*}
The proof is complete by noting that the summation gives the norm of the gradient. 
    \end{proof}

$\bullet$ \textbf{RQSGD:} The optimality gap is $ \frac{\mathfrak{L}~\mathfrak{M}}{T} = \frac{p^{5/2}}{n}{l}^3_{\max}.$ This is because,  $T=n$ as it is one-shot.  Moreover, the factor $\mathfrak{M}$ is calculated using the following fact:
\begin{remark}
    The sub-gradient for RQSGD satisfies $ \EE[\norm{g_t}_2^2] \leq 4p^2 {l}^2_{\max}.$
\end{remark}
 The remark follows from the fact that the sub gradient is of the form  $g_{t} = p (0,\cdots, 0, g_{t,\bfs_l}, 0, \cdots)$ where the index $\bfs_l$ is selected randomly among all the $p$ parameters. Moreover, as it is argued in \cite{HeidariAAAI2022}, each $g_{t,\bfs_l}$ is computed from a Hadamard test and is of the form $ g_{t,\bfs_l}= -2(-1)^z \ell(y, \hat{y})$, where $x\in \{0,1\}$. Therefore, $\abs{ g_{t,\bfs_l}}\leq 2l_{\max}$. This implies that $ \EE[\norm{g_t}_2^2] = p^2\EE[\abs{ g_{t,\bfs_l}}^2]\leq 4p^2 {l}^2_{\max}$.


$\bullet$ \textbf{QSGD:} Here, $T=n$ as this method is also one-shot. From the same argument as in the proof of Theorem \ref{thm:shadow unbiased}, we can show that 
$
\EE[\norm{g_t}_2^2] \leq 4 p 9^k {l}^2_{\max}, 
$
where $k$ is the ansatz locality. Hence, the optimality gap is
$ \frac{\mathfrak{L}~\mathfrak{M}}{T} = \frac{p^{3/2}9^k}{n}{l}^3_{\max}$.